\documentclass[conference]{IEEEtran}
\IEEEoverridecommandlockouts

\usepackage{mathrsfs}
\usepackage{amsmath}
\usepackage{amsfonts}
\usepackage{amsthm}
\usepackage{amssymb}
\usepackage{graphicx}
\usepackage{subfigure}
\usepackage{indentfirst}
\usepackage{array}
\usepackage{cite}
\usepackage{enumerate}
\usepackage{bm}
\usepackage[linesnumbered,ruled,vlined]{algorithm2e}
\usepackage{algorithmic}
\usepackage{multirow}
\usepackage{epstopdf}
\usepackage{verbatim}
\usepackage{stfloats}
\usepackage{color}
\usepackage{caption}
\usepackage{booktabs}
\usepackage{footnote}

\newtheorem{theorem}{\bf{Theorem}}
\newtheorem{lemma}{\bf{Lemma}}

\newtheorem{corollary}{\bf{Corollary}}

\newtheorem{remark}{\bf{Remark}}
\SetKwComment{Comment}{//}{}

\def\BibTeX{{\rm B\kern-.05em{\sc i\kern-.025em b}\kern-.08em
    T\kern-.1667em\lower.7ex\hbox{E}\kern-.125emX}}
    
\usepackage[
    letterpaper,  
    top=0.75in,   
    bottom=1.1in, 
    left=0.673in, 
    right=0.673in,
    twocolumn,    
    columnsep=0.12in 
]{geometry}

\begin{document}

\title{\LARGE{Performance Bounds of Joint Detection with Kalman Filtering and Channel Decoding for Wireless Networked Control Systems}
\thanks{This work is supported in part by the National Key Research and Development Program of China under Grant 2024YFF0509700, in part by the National Natural Science Foundation of China under Grant 62201562, and in part by the Liaoning Provincial Natural Science Foundation of China under Grant 2024--BSBA--51. (\emph{Corresponding Author: Dong Li})}
}

\author{\IEEEauthorblockN{Jinnan Piao\textsuperscript{1}, Dong Li\textsuperscript{1}, Zhibo Li\textsuperscript{1}, Ming Yang\textsuperscript{1}, Xueting Yu\textsuperscript{1}, and Jincheng Dai\textsuperscript{2}}
\IEEEauthorblockA{\textit{\textsuperscript{1}State Key Laboratory of Robotics and Intelligent Systems, Shenyang Institute of Automation, }\\
{\textit  {Chinese Academy of Sciences, Shenyang, China.}} \\
\textit{\textsuperscript{2}Key Laboratory of Universal Wireless Communications, Ministry of Education,} \\
\textit{Beijing University of Posts and Telecommunications, Beijing, China.} \\
Email: \{piaojinnan, lidong, lizhibo, yangming, yuxueting\}@sia.cn, daijincheng@bupt.edu.cn}
\vspace{-0em}}

\maketitle

\begin{abstract}
The joint detection uses Kalman filtering (KF) to estimate the prior probability of control outputs to assist channel decoding. In this paper, we regard the joint detection as maximum a posteriori (MAP) decoding and derive the lower and upper bounds based on the pairwise error probability considering system interference, quantization interval, and weight distribution. We first derive the limiting bounds as the signal-to-noise ratio (SNR) goes to infinity and the system interference goes to zero. Then, we construct an infinite-state Markov chain to describe the consecutive packet losses of the control systems to derive the MAP bounds. Finally, the MAP bounds are approximated as the bounds of the transition probability from the state with no packet loss to the state with consecutive single packet loss. The simulation results show that the MAP performance of $\left(64,16\right)$ polar code and 16-bit CRC coincides with the limiting upper bound as the SNR increases and has $3.0$dB performance gain compared with the normal approximation of the finite block rate at block error rate $10^{-3}$.
\end{abstract}

\begin{IEEEkeywords}
MAP decoding, Kalman filtering, channel decoding, wireless networked control systems, polar codes.
\end{IEEEkeywords}

\section{Introduction}

Wireless networked control systems (WNCSs) are control systems  with the components, i.e., controllers, sensors and actuators, distributed and connected via wireless communication channels, where the control stability and performance tightly relate to the transmission error \cite{NCS_Survey, WNCS_Survey}.

Kalman filtering (KF) plays a fundamental role in estimating system states when transmission error occurs \cite{KalmanLQR}.
Sinopoli \emph{et al.} \cite{KalmanBernoulliOriginal} propose the KF with intermittent observations
and prove that existing a critical value makes the expected estimation error covariance is bounded when the transmission process follows Bernoulli distribution and the error probability is less than the critical value.
Then, the probability bounds of the estimation error covariance are proposed in \cite{5398830}.
Currently, KF with intermittent observations is widely used in WNCSs, such as robots \cite{11007832}, vehicles \cite{10387409} and unmanned aerial vehicles \cite{9453790}.

Channel codes can effectively mitigate transmission error \cite{LinShuBook} and a simple on-off error control coding scheme can improve control quality \cite{WNCSchannelCodingOnOff}. To ensure a wide range of application scenarios, channel codes are generally designed assuming that the transmitted bits obey uniform distribution and no prior information is considered.
To utilize the prior information estimated by KF, a maximum a posteriori (MAP) receiver for each element of the system state is proposed in \cite{WNCSchannelCodingCRC} and an iterative joint detection between KF and channel decoding with LDPC codes is proposed in \cite{arxivIJDLDPC}, which exhibits potential in optimizing the block error rate (BLER) performance of communications and the root mean square error (RMSE) performance of controls with the prior information.

In this paper, we regard the joint detection with KF and channel decoding as MAP decoding, where the KF is used to estimate the probability density of system outputs to calculate the prior probability.
We first derive the limiting lower and upper bounds of MAP decoding with the pairwise error probability considering control system interference, quantization interval, and codeword weight distribution, as the signal-to-noise ratio (SNR) goes to infinity and the system interference goes to zero.
Then, we construct an infinite-state Markov chain to characterize the MAP performance with the stationary distribution to derive the MAP bounds, where the state ${S}_i$ is the event that the number of consecutive packet losses is $i$, $i=0,1,\cdots$.
Finally, we establish a relationship between the MAP bounds and upper bounds of the transition probabilities and approximate the MAP bounds as the bounds of the transition probability from
state $S_0$ to state $S_1$.
The simulation results show that the MAP performance coincides with the limiting upper bound as the SNR increases and outperforms the normal approximation of the finite block rate, and the RMSE of the systems with MAP decoding is lower than that with maximum likelihood (ML) decoding, which demonstrates the advantage of jointly designing control and communication systems with prior information.

\emph{Notation Conventions}: The lowercase letters, e.g., $x$, are used to denote scalars. The bold lowercase letters, e.g., ${\mathbf{x}}$, are used to denote vectors.
The notation $x_i$ denotes the $i$-th element of ${\mathbf{x}}$.
Sets are denoted by calligraphic characters, e.g., $\mathcal{X}$, and the notation $|\mathcal{X}|$ denotes the cardinality of $\mathcal{X}$.
The bold capital letters, e.g., $\mathbf{X}$, are used to denote matrices.
${\mathbb{R}}$ represents the real number field.

\section{System Model}

We consider a discrete linear time-invariant system as
\begin{equation*}
\begin{aligned}
  \mathbf{x}\left[ k+1 \right]&=\mathbf{Ax}\left[ k \right]+\mathbf{Bu}\left[ k \right]+\mathbf{w}\left[ k \right],  \\
 \mathbf{y}\left[ k \right]&=\mathbf{Cx}\left[ k \right]+\mathbf{v}\left[ k \right],
\end{aligned}
\end{equation*}
where $\mathbf{x}\left[ k \right] \in {{\mathbb{R}}^{{{N}_{x}}}}$, $\mathbf{u}\left[ k \right] \in {{\mathbb{R}}^{{{N}_{u}}}}$, and $\mathbf{y}\left[ k \right]\in {{\mathbb{R}}^{{{N}_{y}}}}$ are the state vector, the input vector, and the output vector at time index $k$, respectively. $\mathbf{A}\in {{\mathbb{R}}^{{{N}_{x}}\times {{N}_{x}}}}$, $\mathbf{B}\in {{\mathbb{R}}^{{{N}_{x}}\times {{N}_{u}}}}$, and $\mathbf{C}\in {{\mathbb{R}}^{{{N}_{y}}\times {{N}_{x}}}}$ are the known system parameter matrices with appropriate dimensions. $\mathbf{w}\left[ k \right]\in {{\mathbb{R}}^{{{N}_{x}}}}$ and $\mathbf{v}\left[ k \right]\in {{\mathbb{R}}^{{{N}_{y}}}}$ are Gaussian noises with zero means and covariance matrices ${\mathbf W}$ and ${\mathbf V}$, respectively.

Each element of $\mathbf{y}\left[ k \right]$ is quantized by an $n$-bit uniform quantizer with the quantization range $\left[ -Z, Z \right)$.
Define $\alpha \left( x \right)\triangleq {\hat z}_l = 0.5\left( {{z}_{l}}+{{z}_{l+1}} \right)$ as the midpoint of quantized interval $x\in \left[ {{z}_{l}},{{z}_{l+1}} \right)$,
$\beta \left( x \right)\triangleq l$ as the index of quantized interval with $z_0 = -Z-\Delta/2$, $z_l = z_{l-1} + \Delta$ and $\Delta ={Z}/{ {2}^{n-1} }$, and ${\mathbf b}^m = \left[b_{n-1}^m, \cdots, b_1^m, b_0^m\right]^T$ as the binary representation of $m$ with the most and the least significant bits $b_{n-1}^m$ and $b_0^m$.
The $N_y$-length quantized vector is
\begin{equation*}
  \mathbf{q}\left[ k \right]= \alpha\left({\bf y}\left[k\right]\right) = \left[ \alpha \left( {{y}_{1}}\left[ k \right] \right),\alpha \left( {{y}_{2}}\left[ k \right] \right),\cdots ,\alpha \left( {{y}_{{{N}_{y}}}}\left[ k \right] \right) \right]^T
\end{equation*}
and the $\left(N_yn\right)$-bit vector sent to the communication layer is
\begin{equation*}
\mathbf{b}\left[ k \right]= {\bf b}^{\beta\left({\bf y}\left[k\right]\right)} =\left[ \left({{\mathbf{b}}^{\beta \left( {{y}_{1}}\left[ k \right] \right)}}\right)^T,\cdots ,\left({{\mathbf{b}}^{\beta \left( {{y}_{{{N}_{y}}}}\left[ k \right] \right)}}\right)^T \right]^T.
\end{equation*}

In the communication layer, for a channel code $\left(N, K = N_yn\right)$ with code rate $R = K/N$, the codeword is $\mathbf{c}\left[k\right] = \mathbf{G}\mathbf{b}\left[k\right]$, where $\mathbf{G}$ is the generator matrix.
The codeword ${\bf c}\left[k\right]$ is modulated into the transmitted vector by binary phase shift keying (BPSK), i.e., ${\bf t}\left[k\right] = {\bf 1} - 2{\bf c}\left[k\right]$.
The received vector is $\mathbf{r}\left[k\right] = \mathbf{t}\left[k\right] + \mathbf{n}\left[k\right]$,
where $n_i\left[k\right]$ is i.i.d. additive white Gaussian noise (AWGN) with zero mean and variance $N_0/2$.

After the MAP detection, the estimation $\hat{\mathbf{q}}\left[ k \right]$ of $\mathbf{q}\left[ k \right]$ is obtained.
Given the set ${\mathcal F}\left[ k \right] = \left\{\hat{\mathbf q}\left[ 1 \right],\cdots, \hat{\mathbf q}\left[ k \right], \gamma \left[ 1 \right],\cdots,\gamma \left[ k \right]\right\}$,
the KF with intermittent observations is used to calculate the estimated system states $\hat{\mathbf x}\left[ k \right]$ and covariance matrix $\hat{\mathbf{P}}\left[ k \right]$ as
\begin{equation*}
\begin{aligned}
{{\hat{\bf x}}_ - }\left[ k \right] & = \mathbb{E}\left(\left.\mathbf{x}\left[ k \right]\right| {\mathcal F}\left[ k-1 \right] \right)
    = {{{\bf A}\hat{\bf x}}}\left[ {k - 1} \right] + {\bf{Bu}}\left[ {k - 1} \right]\\
{{\hat{\bf e}}_ - }\left[ k \right] &= {{\mathbf{x}}\left[ k \right] - {{{\hat{\mathbf x}}}_ - }\left[ k \right]}\\
{{{\hat{\bf P}}}_ - }\left[ k \right] & =  \mathbb{E}\left( {\left. {{{\hat{\bf e}}_-^T}\left[ k \right]{{\hat{\bf e}}_ - }\left[ k \right]} \right|\mathcal{F}\left[ {k - 1} \right]} \right)\\
    &= {{{\bf A}\hat{\bf P}}}\left[ {k - 1} \right]{{\bf{A}}^T} + {\bf{W}}\\
{\bf{K}}\left[ k \right] &= {{{\hat{\bf P}}}_ - }\left[ k \right]{{\bf{C}}^T}{( {{\bf{C}}{{{\hat{\bf P}}}_ - }\left[ k \right]{{\bf{C}}^T} + {\bf{V}}} )^{ - 1}}\\
{\hat{\bf x}}\left[ k \right] &= \mathbb{E}\left(\left.\mathbf{x}\left[ k \right]\right| {\mathcal F}\left[ k \right] \right)\\
    &= {{{\hat{\bf x}}}_ - }\left[ k \right] + \gamma \left[ k \right]{\bf{K}}\left[ k \right]( {{\hat{\bf q}}\left[ k \right] - {\bf{C}}{{{\hat{\bf x}}}_ - }\left[ k \right]} )\\
{\hat{\bf e}}\left[ k \right] &= {{\mathbf{x}}\left[ k \right] - {{{\hat{\mathbf x}}}}\left[ k \right]} \\
{\hat{\bf P}}\left[ k \right] & =  \mathbb{E}\left( {\left. {{{\hat{\bf e}}^T}\left[ k \right]{{\hat{\bf e}}}\left[ k \right]} \right|\mathcal{F}\left[ {k} \right]} \right)
    = \left( {{\bf{I}} - \gamma \left[ k \right]{\bf{K}}\left[ k \right]{\bf{C}}} \right){{{\hat{\bf P}}}_ - }\left[ k \right],
\end{aligned}
\end{equation*}
where $\gamma\left[k\right] = 1$ if the decoded bits are correct, and $\gamma\left[k\right] = 0$ otherwise.
Then, the controller uses $\hat{\mathbf{x}}\left[ k \right]$ and the reference state vector ${{\mathbf{x}}_{\rm{ref}}}\left[ k \right] \in {\mathbb R}^{N_x}$ to calculate the input vector as
$\mathbf{u}\left[ k \right] = \mathbf{K}_{\rm con}\left( {{\mathbf{x}}_{\rm{ref}}}\left[ k \right]- \hat{\mathbf{x}}\left[ k \right]\right)$,
where $\mathbf{K}_{\rm con} \in {\mathbb R}^{N_u \times N_x}$ is the controller gain matrix.

We define the functions as
$h\left(\bf X\right)  \triangleq  {\bf A}{\bf X}{\bf A}^T + {\bf W}$,
$g\left(\bf X\right)  \triangleq {\bf C}{\bf X}{\bf C}^T + {\bf V}$, and
$\tilde g\left( {\mathbf{X}} \right)   \triangleq {\mathbf{X}} - {\mathbf{X}}{{\mathbf{C}}^T}{\left( {{\mathbf{CX}}{{\mathbf{C}}^T} + {\mathbf{V}}} \right)^{ - 1}}{\mathbf{CX}}$.
For functions $f_1$ and $f_2$, 
$
{f_1} \circ {f_2}  \triangleq f_1\left(f_2\left(\bf X\right)\right),
$
and 
$
f_1^a\left(\bf X\right)  \triangleq \underbrace {{f_1} \circ {f_1} \circ  \cdots  \circ {f_1}}_{a~\rm{times}}\left( {\bf{X}} \right).
$

\section{Limiting Bounds of MAP Decoding}

In this section, the MAP decoding criterion is provided and the limiting bounds of the performance are derived. To simplify the representation, we omit the time index $\left[k\right]$.

Since ${{{\hat{\bf x}}}_ - }$ and ${{{\hat{\bf P}}}_ - }$ are the predicted system states and covariance matrix, respectively, we have the predicted output
\begin{equation}
{\bm\mu} = \mathbb{E}\left(\left.\mathbf{y}\right| {\mathcal F}\left[ k-1 \right] \right) = \mathbb{E}\left(\left.\mathbf{Cx}\right| {\mathcal F}\left[ k-1 \right] \right) = \mathbf{C}{{{\hat{\bf x}}}_ - }
\end{equation}
with the corresponding covariance matrix
\begin{equation}
\mathbf{\Sigma} = \mathbb{E}\left(\left.\left(\mathbf{y}-\bm{\mu}\right)^T\left(\mathbf{y}-\bm{\mu}\right) \right| {\mathcal F}\left[ k-1 \right] \right) = g({\hat{\bf P}}_ - ).
\end{equation}
Hence, $\mathbf{y}$ obeys Gaussian distribution, i.e., $\mathbf{y} \sim \mathcal{N}({\bm\mu}, \mathbf{\Sigma})$, with
\begin{equation*}
p\left(\mathbf{y}\right) =  \frac{1}{{{{\left( {2\pi } \right)}^{\frac{{{N_y}}}{2}}}{{\left(\det \mathbf{\Sigma} \right)}^{\frac{1}{2}}}}}\exp \left( { - \frac{1}{2}{{\left( {{\bf{y}} - {\bm\mu}} \right)}^T}\mathbf{\Sigma}^{ - 1}\left( {{\bf{y}} - {\bm\mu}} \right)} \right).
\end{equation*}
With this, the MAP detection criterion is
\begin{equation}\label{EqMAPdetectionCriterion}
\begin{aligned}
{\hat {\mathbf b}} &= \mathop {\arg \max }\limits_{{\bf{b}} \in {{\left\{ {0,1} \right\}}^{{N_y}n}}} \Pr \left( {\left. {\bf{b}} \right|{\bf{r}},{\bm\mu}, \mathbf{\Sigma}} \right) \\
&= \mathop {\arg \max }\limits_{{\bf{b}} \in {{\left\{ {0,1} \right\}}^{{N_y}n}}} \Pr \left( \left. {\bf{r}} \right| {\bf{b}} \right) \Pr \left( {\left. {\bf{b}} \right|{\bm\mu}, \mathbf{\Sigma}} \right) \\
& = \mathop {\arg \max }\limits_{{\bf{b}} \in {{\left\{ {0,1} \right\}}^{{N_y}n}}} 2{\mathbf r}^T{\mathbf t}+{N_0}\ln\left(P_{{\bm\mu}, \mathbf{\Sigma}}^{{\bf b}}\right),
\end{aligned}
\end{equation}
where
$
P_{{\bm\mu}, \mathbf{\Sigma}}^{{\bf b}} =
\Pr \left( {\left. {\bf{b}} \right|{\bm\mu}, \mathbf{\Sigma}} \right) = \int\nolimits_{{\bf{y}} \in \left\{ {\left. {\bf{y}} \right| {{\bf{b}}^{\beta \left( {\bf{y}} \right)} = {\bf{b}}}} \right\}} {p\left( {\bf{y}} \right)\mathrm{d}{\bf{y}}}.
$

The MAP error probability given ${\bm\mu}$ and ${\mathbf{\Sigma}}$ is
\begin{equation}
P_{{\bm\mu}, \mathbf{\Sigma}}^{{\rm{MAP}}} = \Pr \left( {\left. {{{\bf{b}}_t} \ne \hat {\bf{b}}} \right|{\bm\mu}, \mathbf{\Sigma}} \right)
\end{equation}
and the MAP error probability $P^{{\rm{MAP}}}$ is calculated as
\begin{equation}\label{EqMAPcalprocess}
\begin{aligned}
P^{{\rm{MAP}}} & = \Pr \left( {{{\bf{b}}_t} \ne \hat {\bf{b}}} \right) = \int\nolimits_{{\bm\mu}, \mathbf{\Sigma}} P_{{\bm\mu}, \mathbf{\Sigma}}^{{\rm{MAP}}} p\left({{\bm\mu}, \mathbf{\Sigma}}\right) \mathrm{d}{\bm\mu} \mathrm{d}\mathbf{\Sigma}  \\
&= \int\nolimits_{\mathbf{\Sigma}} \left( \int\nolimits_{{\bm\mu}} P_{{\bm\mu}, \mathbf{\Sigma}}^{{\rm{MAP}}} p\left(\left.{\bm\mu}\right| \mathbf{\Sigma}\right) \mathrm{d}{\bm\mu} \right) p\left(\mathbf{\Sigma}\right) \mathrm{d}\mathbf{\Sigma}
\end{aligned}
\end{equation}

Then, the pairwise error probability (PEP) $P_{{\bf b}_t, {\bm\mu}, \mathbf{\Sigma}}^{{\bf b}_t \ne {\bf b}_e}$, the bounds of $P_{{\bm\mu}, \mathbf{\Sigma}}^{{\rm{MAP}}}$, and the limiting bounds of $P^{{\rm{MAP}}}$ are provided in Lemma \ref{LemmaPEP}, Lemma \ref{LemmaMAPbounds}, and Theorem \ref{TheoremMAPlimit}, respectively.

\begin{lemma}\label{LemmaPEP}
Given the transmitted bit vector ${\mathbf b}_t$, the estimated bit vector ${\mathbf b}_e$ and the Hamming weight ${d_{{{\bf{c}}_t},{{\bf{c}}_e}}}$ between ${{{\bf{c}}_t}}$ and ${{{\bf{c}}_e}}$, the PEP $P_{{\bf b}_t, {\bm\mu}, \mathbf{\Sigma}}^{{\bf b}_t \ne {\bf b}_e}$ is
\begin{equation}\label{EqPEP}
\begin{aligned}
&P_{{\bf b}_t, {\bm\mu}, \mathbf{\Sigma}}^{{\bf b}_t \ne {\bf b}_e} = {\Pr}\left( {\left. {{{\bf{b}}_t} \ne {{\bf{b}}_e}} \right|{{\bf{b}}_t},{\bm\mu}, \mathbf{\Sigma}} \right)\\
&= Q\left( {\sqrt {\frac{{2{d_{{{\bf{c}}_t},{{\bf{c}}_e}}}}}{{{N_0}}}}  + \sqrt {\frac{{{N_0}}}{{8{d_{{{\bf{c}}_t},{{\bf{c}}_e}}}}}} \ln \left( {\frac{{P_{{\bm\mu}, \mathbf{\Sigma}}^{{\bf b}_t}}}{{P_{{\bm\mu}, \mathbf{\Sigma}}^{{\bf b}_e}}}} \right)} \right)
\end{aligned}
\end{equation}
\end{lemma}

\begin{proof}
Given ${\mathbf r}_t = {\mathbf t}_t + {\mathbf n}$, we have
\begin{equation}\label{EqPEP1}
\begin{aligned}
&P_{{\bf b}_t, {\bm\mu}, \mathbf{\Sigma}}^{{\bf b}_t \ne {\bf b}_e}
 = \Pr \left( {2{{\bf{r}}_t^T}\left( {{{\bf{t}}_e} - {{\bf{t}}_t}} \right) - {N_0}\ln \left( {\frac{{P_{{\bm\mu}, \mathbf{\Sigma}}^{{\bf b}_t}}}{{P_{{\bm\mu}, \mathbf{\Sigma}}^{{\bf b}_e}}}} \right) > 0} \right) \\
 &= \Pr \left( {{{\bf{n}}^T}\left( {{{\bf{t}}_e} - {{\bf{t}}_t}} \right) > 2{d_{{{\bf{c}}_t},{{\bf{c}}_e}}} + \frac{{{N_0}}}{2}\ln \left( {\frac{{P_{{\bm\mu}, \mathbf{\Sigma}}^{{\bf b}_t}}}{{P_{{\bm\mu}, \mathbf{\Sigma}}^{{\bf b}_e}}}} \right)} \right).
\end{aligned}
\end{equation}
Then, since ${{\bf{n}}^T}\left( {{{\bf{t}}_e} - {{\bf{t}}_t}} \right) \sim {\mathcal N}\left(0, 2dN_0\right)$, \eqref{EqPEP} is derived.
\end{proof}

\begin{lemma}\label{LemmaMAPbounds}
$P_{{\bm\mu}, \mathbf{\Sigma}}^{{\rm{MAP}}}$ is bounded as
$
P_{{\bm\mu}, \mathbf{\Sigma}}^{\mathrm{LB}} \le P_{{\bm\mu}, \mathbf{\Sigma}}^{{\rm{MAP}}} \le P_{{\bm\mu}, \mathbf{\Sigma}}^{\mathrm{UB}}
$, where
\begin{equation}
P_{{\bm\mu}, \mathbf{\Sigma}}^{\mathrm{LB}} = \sum\limits_{{\bf{b}}_t} {P_{{\bm\mu}, \mathbf{\Sigma}}^{{\bf b}_t}\mathop {\max }\limits_{{{\bf{b}}_e} \in {{\left\{ {0,1} \right\}}^{{N_y}n}}\backslash {{\bf{b}}_t}} P_{{\bf b}_t, {\bm\mu}, \mathbf{\Sigma}}^{{\bf b}_t \ne {\bf b}_e}},
\end{equation}
\begin{equation}
P_{{\bm\mu}, \mathbf{\Sigma}}^{\mathrm{UB}} = \sum\limits_{{\bf{b}}_t} {P_{{\bm\mu}, \mathbf{\Sigma}}^{{\bf b}_t}\sum\limits_{{{\bf{b}}_e} \in {{\left\{ {0,1} \right\}}^{{N_y}n}}\backslash {{\bf{b}}_t}} P_{{\bf b}_t, {\bm\mu}, \mathbf{\Sigma}}^{{\bf b}_t \ne {\bf b}_e}}.
\end{equation}
\end{lemma}
\begin{proof}
Obviously, due to
\begin{equation}
\begin{aligned}
\mathop {\max }\limits_{{{\bf{b}}_e} \in {{\left\{ {0,1} \right\}}^{{N_y}n}}\backslash {{\bf{b}}_t}} P_{{\bf b}_t, {\bm\mu}, \mathbf{\Sigma}}^{{\bf b}_t \ne {\bf b}_e} &\le
{\Pr}\left( {\left. {{{\bf{b}}_t} \ne {\hat{\bf{b}}}} \right|{{\bf{b}}_t},{\bm\mu}, \mathbf{\Sigma}} \right)\\
 &\le \sum\limits_{{{\bf{b}}_e} \in {{\left\{ {0,1} \right\}}^{{N_y}n}}\backslash {{\bf{b}}_t}} P_{{\bf b}_t, {\bm\mu}, \mathbf{\Sigma}}^{{\bf b}_t \ne {\bf b}_e},
\end{aligned}
\end{equation}
Lemma \ref{LemmaMAPbounds} is proved.
\end{proof}

\begin{theorem}\label{TheoremMAPlimit}
Assuming ${{\mathbf{x}}_{\rm{ref}}}$ is fixed and the control system is stable, as the SNR $1/N_0$ goes to infinity and ${{\bf{P}}_\infty }$ goes to zero, the limiting performance is
\begin{equation}\label{EqMAPlimitation}
\lim\limits_{\frac{1}{N_0} \rightarrow \infty, {{\bf{P}}_\infty } \rightarrow 0} P^{{\rm{MAP}}} = P_{{\bm\mu}_{\rm{ref}}, \mathbf{\Sigma}_{\infty}}^{{\rm{MAP}}},
\end{equation}
which is bounded as
\begin{equation}\label{EqMAPlimitationBounds}
P_{{\bm\mu}_{\rm{ref}}, \mathbf{\Sigma}_{\infty}}^{\mathrm{LB}} \le P_{{\bm\mu}_{\rm{ref}}, \mathbf{\Sigma}_{\infty}}^{{\rm{MAP}}} \le P_{{\bm\mu}_{\rm{ref}}, \mathbf{\Sigma}_{\infty}}^{\mathrm{UB}},
\end{equation}
where ${\bm\mu}_{\rm{ref}} = {\bf C}{{\mathbf{x}}_{\rm{ref}}}$, $\mathbf{\Sigma}_{\infty} = g\circ h\left({{\bf{P}}_\infty }\right)$, ${{\bf{P}}_\infty } = \tilde g\left( {\hat{\bf P}}_{-}\left[ \infty \right] \right)$ and ${\hat{\bf P}}_{-}\left[ \infty \right]$ is the result of standard algebraic Riccati equation as
$
{\hat{\bf P}}_{-}\left[ \infty \right] =h\circ \tilde g \left({\hat{\bf P}}_{-}\left[ \infty \right]\right)
$.
\end{theorem}
\begin{proof}
As the SNR $1/N_0$ goes to infinity, we have
\begin{equation}\label{EqLemma3Lim1}
\begin{aligned}
\lim\limits_{\frac{1}{N_0} \rightarrow \infty} \Pr\left(\gamma = 1\right) = 1 &\Rightarrow \lim\limits_{\frac{1}{N_0} \rightarrow \infty} \Pr\left({\hat{\bf P}} = {{\bf{P}}_\infty }\right) = 1\\
&\Rightarrow \lim\limits_{\frac{1}{N_0} \rightarrow \infty} \Pr\left(\mathbf{\Sigma} = \mathbf{\Sigma}_{\infty}\right) = 1.
\end{aligned}
\end{equation}
Then, as ${{\bf{P}}_\infty } \rightarrow 0$, we have
\begin{equation}\label{EqLemma3Lim2}
\lim\limits_{\frac{1}{N_0} \rightarrow \infty, {{\bf{P}}_\infty } \rightarrow 0} \Pr\left(\left.{\bm\mu} = {\bm\mu}_{\rm{ref}}\right|\mathbf{\Sigma}_{\infty}\right) = 1.
\end{equation}
With \eqref{EqMAPcalprocess}, \eqref{EqLemma3Lim1} and \eqref{EqLemma3Lim2}, we obtain \eqref{EqMAPlimitation}. With Lemma \ref{LemmaMAPbounds} and \eqref{EqMAPlimitation}, we obtain \eqref{EqMAPlimitationBounds}.
\end{proof}

\section{Markov Chain of Control System and Bounds}

In this section, we use an infinite-state Markov chain to describe the control system with MAP decoding to derive the corresponding bounds.

The infinite-state Markov chain is shown in Fig. \ref{FigMarkov}, where the state ${S}_i$ is the event that the number of consecutive packet losses is $i$, $i=0,1,\cdots$. Hence, if a new packet losses, the state moves from $S_i$ to $S_{i+1}$, i.e., the number of consecutive packet losses increases from $i$ to $i+1$, and if a new packet arrives, the state moves from $S_i$ to $S_{0}$.
With this, the transition probabilities are $ P_{S_i}^{S_{i+1}} = \Pr\left(\left.S_{i+1}\right|S_i\right)$ and $P_{S_i}^{S_{0}} =\Pr\left(\left.S_{0}\right|S_i\right) =  1 - P_{S_i}^{S_{i+1}}$, where
$
P_{S_i}^{S_{i+1}} = P_{S_i}^{{\rm{MAP}}} = \Pr( {{{\bf{b}}_t} \ne \hat {\bf{b}}} |S_i).
$

Defining $\pi\left(S_i\right)$ as the stationary distribution of $S_i$, $i=0,1,\cdots$, we have
\begin{equation}
\begin{aligned}
\pi \left( {{S_0}} \right) &= \sum\limits_{i = 0}^\infty  {P_{{S_i}}^{{S_0}}\pi \left( {{S_i}} \right)}  = \sum\limits_{i = 0}^\infty  {\Pr \left( {\left. {{{\mathbf{b}}_t} = \widehat {\mathbf{b}}} \right|{S_i}} \right)\pi \left( {{S_i}} \right)}  \\
&= \Pr \left( {{{\mathbf{b}}_t} = \widehat {\mathbf{b}}} \right),
\end{aligned}
\end{equation}
and the MAP performance is
\begin{equation}\label{EqPMAPSi}
P^{{\rm{MAP}}} = \Pr \left( {{{\mathbf{b}}_t} \ne \widehat {\mathbf{b}}} \right) = 1 - \pi \left( {{S_0}} \right) = \frac{{P_{{S_0}}^{{S_1}}\eta }}{{1 + P_{{S_0}}^{{S_1}}\eta }},
\end{equation}
where
$
\pi \left( {{S_0}} \right){\rm{ }} = {\left( {1 + P_{{S_0}}^{{S_1}}\eta } \right)^{ - 1}}$,
$\pi \left( {{S_{i+1}}} \right) = P_{{S_{i}}}^{{S_{i+1}}}\pi \left( {{S_{i}}} \right),i = 0,1, \cdots$, and
$\eta = 1 + \sum\nolimits_{i = 1}^\infty  {\left( {\prod\nolimits_{j = 1}^i {P_{{S_{j}}}^{{S_{j+1}}}} } \right)}.
$
We derive the bounds of $P_{S_i}^{S_{i+1}}$ to obtain the bounds of \eqref{EqPMAPSi}.

\begin{figure}[t]
\setlength{\abovecaptionskip}{0.cm}
\setlength{\belowcaptionskip}{-0.cm}
  \centering{\includegraphics[scale=0.75]{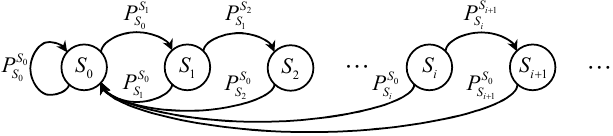}}
  \caption{Markov chain of control system with MAP decoding.}\label{FigMarkov}
  \vspace{-0em}
\end{figure}

To derive the bounds, we introduce the Loewner partial order \cite{MatrixAnalysis} in Lemma \ref{LemmaLoewner} and provide the bounds of $P_{{\bm\mu}, \mathbf{\Sigma}}^{{\bf b}}$, $P_{{\bf b}_t, {\bm\mu}, \mathbf{\Sigma}}^{{\bf b}_t \ne {\bf b}_e}$, $P_{{\bm\mu}, \mathbf{\Sigma}}^{\mathrm{LB}}$ and $P_{{\bm\mu}, \mathbf{\Sigma}}^{\mathrm{UB}}$ given $\underline{\bf\Sigma} \le {\bf{\Sigma }} \le \overline{\bf{\Sigma }}$ in Lemma \ref{LemmaBoundsGivenSigmaBound}.
\begin{lemma}\label{LemmaLoewner}
For matrices ${\mathbf X}, {\mathbf Y} \in {{\mathbb{R}}^{{a}\times {a}}}$, we write ${\mathbf X} \le {\mathbf Y}$ if ${\mathbf X}$ and ${\mathbf Y}$ are symmetric and ${\mathbf Y}-{\mathbf X}$ is positive semidefinite. Then, we have  ${\mathbf Y}^{-1} \le {\mathbf X}^{-1}$, $0\le\det{\mathbf X} \le \det{\mathbf Y}$, and ${\mathbf S}^T{\mathbf X}{\mathbf S} \le {\mathbf S}^T{\mathbf Y}{\mathbf S}$ if ${\mathbf S}\in {{\mathbb{R}}^{a\times b}}$.
\end{lemma}

\begin{lemma}\label{LemmaBoundsGivenSigmaBound}
Given the range of ${\bf \Sigma}$ as $\underline{\bf\Sigma} \le {\bf{\Sigma }} \le \overline{\bf{\Sigma }}$, $P_{{\bm\mu}, \mathbf{\Sigma}}^{{\bf b}}$ is bounded as
\begin{equation}\label{EqLemma41}
\underline{P_{{\bm\mu}, \mathbf{\Sigma}}^{{\bf b}}} \le P_{{\bm\mu}, \mathbf{\Sigma}}^{{\bf b}} \le \overline{P_{{\bm\mu}, \mathbf{\Sigma}}^{{\bf b}}},
\end{equation}
$P_{{\bf b}_t, {\bm\mu}, \mathbf{\Sigma}}^{{\bf b}_t \ne {\bf b}_e}$ is bounded as
\begin{equation}\label{EqLemma42}
\underline{P_{{\bf b}_t, {\bm\mu}, \mathbf{\Sigma}}^{{\bf b}_t \ne {\bf b}_e}} \le P_{{\bf b}_t, {\bm\mu}, \mathbf{\Sigma}}^{{\bf b}_t \ne {\bf b}_e} \le \overline{P_{{\bf b}_t, {\bm\mu}, \mathbf{\Sigma}}^{{\bf b}_t \ne {\bf b}_e}},
\end{equation}
the lower bound of $P_{{\bm\mu}, \mathbf{\Sigma}}^{\mathrm{LB}}$ is
\begin{equation}\label{EqLemma43}
P_{{\bm\mu}, \mathbf{\Sigma}}^{\mathrm{LB}}\ge\underline{P_{{\bm\mu}, \mathbf{\Sigma}}^{\mathrm{LB}}} =
\sum\limits_{{\bf{b}}_t} {\underline{P_{{\bm\mu}, \mathbf{\Sigma}}^{{\bf b}_t}}\mathop {\max }\limits_{{{\bf{b}}_e} \in {{\left\{ {0,1} \right\}}^{{N_y}n}}\backslash {{\bf{b}}_t}} \underline{P_{{\bf b}_t, {\bm\mu}, \mathbf{\Sigma}}^{{\bf b}_t \ne {\bf b}_e}}
}
\end{equation}
and the upper bound of $P_{{\bm\mu}, \mathbf{\Sigma}}^{\mathrm{UB}}$ is
\begin{equation}\label{EqLemma44}
P_{{\bm\mu}, \mathbf{\Sigma}}^{\mathrm{UB}} \le \overline{P_{{\bm\mu}, \mathbf{\Sigma}}^{\mathrm{UB}}} =
\sum\limits_{{\bf{b}}_t} {\overline{P_{{\bm\mu}, \mathbf{\Sigma}}^{{\bf b}_t}}\sum\limits_{{{\bf{b}}_e} \in {{\left\{ {0,1} \right\}}^{{N_y}n}}\backslash {{\bf{b}}_t}} \overline{P_{{\bf b}_t, {\bm\mu}, \mathbf{\Sigma}}^{{\bf b}_t \ne {\bf b}_e}}},
\end{equation}
where
\begin{equation*}
\begin{aligned}
\underline{P_{{\bm\mu}, \mathbf{\Sigma}}^{{\bf b}}} &= {\left( {\frac{{\det \underline {\bf{\Sigma }} }}{{\det \overline {\bf{\Sigma }} }}} \right)^{\frac{1}{2}}}P_{\bm\mu ,\underline {\bf{\Sigma }} }^{\bf{b}},
\overline{P_{{\bm\mu}, \mathbf{\Sigma}}^{{\bf b}}} = {\left( {\frac{{\det \overline {\bf{\Sigma }} }}{{\det \underline {\bf{\Sigma }} }}} \right)^{\frac{1}{2}}}P_{\bm\mu ,\overline {\bf{\Sigma }} }^{\bf{b}}\\
\underline{P_{{\bf b}_t, {\bm\mu}, \mathbf{\Sigma}}^{{\bf b}_t \ne {\bf b}_e}} &= Q\left( {\sqrt {\frac{{2{d_{{{\bf{c}}_t},{{\bf{c}}_e}}}}}{{{N_0}}}}  + \sqrt {\frac{{{N_0}}}{{8{d_{{{\bf{c}}_t},{{\bf{c}}_e}}}}}} \ln \left(\frac{{P_{\bm\mu ,\overline {\bf{\Sigma }} }^{{{\bf{b}}_t}}\det \overline {\bf{\Sigma }} }}{{P_{\bm\mu ,\underline {\bf{\Sigma }} }^{{{\bf{b}}_e}}\det \underline {\bf{\Sigma }} }}  \right)} \right)\\
\overline{P_{{\bf b}_t, {\bm\mu}, \mathbf{\Sigma}}^{{\bf b}_t \ne {\bf b}_e}} &= Q\left( {\sqrt {\frac{{2{d_{{{\bf{c}}_t},{{\bf{c}}_e}}}}}{{{N_0}}}}  + \sqrt {\frac{{{N_0}}}{{8{d_{{{\bf{c}}_t},{{\bf{c}}_e}}}}}} \ln \left(\frac{{P_{\bm\mu ,\underline {\bf{\Sigma }} }^{{{\bf{b}}_t}}\det \underline {\bf{\Sigma }} }}{{P_{\bm\mu ,\overline {\bf{\Sigma }} }^{{{\bf{b}}_e}}\det \overline {\bf{\Sigma }} }}  \right)} \right)
\end{aligned}
\end{equation*}
\end{lemma}
\begin{proof}
With Lemma \ref{LemmaLoewner} and $\underline{\bf\Sigma} \le {\bf{\Sigma }} \le \overline{\bf{\Sigma }}$, we have $\underline{p\left(\mathbf{y}\right)} \le p\left(\mathbf{y}\right) \le
\overline{p\left(\mathbf{y}\right)}$ with
\begin{equation*}
\begin{aligned}
\underline{p\left(\mathbf{y}\right)} &= \frac{1}{{{{\left( {2\pi } \right)}^{\frac{{{N_y}}}{2}}}{{\left( {\det \overline {\bf{\Sigma }} } \right)}^{\frac{1}{2}}}}}\exp \left( { - \frac{1}{2}{{\left( {{\bf{y}} - \bm\mu } \right)}^T}{{\underline {\bf{\Sigma }} }^{ - 1}}\left( {{\bf{y}} - \bm\mu } \right)} \right) \\
\overline{p\left(\mathbf{y}\right)} &= \frac{1}{{{{\left( {2\pi } \right)}^{\frac{{{N_y}}}{2}}}{{\left( {\det \underline {\bf{\Sigma }} } \right)}^{\frac{1}{2}}}}}\exp \left( { - \frac{1}{2}{{\left( {{\bf{y}} - \bm\mu } \right)}^T}{{\overline {\bf{\Sigma }} }^{ - 1}}\left( {{\bf{y}} - \bm\mu } \right)} \right)
\end{aligned}
\end{equation*}
Then, we have
\begin{align}
\underline{P_{{\bm\mu}, \mathbf{\Sigma}}^{{\bf b}}} =
\int\nolimits_{{\bf{y}} \in \left\{ {\left. {\bf{y}} \right|{\bf{b}} = {{\bf{b}}^{\beta \left( {\bf{y}} \right)}}} \right\}} \underline{p\left( {\bf{y}} \right)}\mathrm{d}{\bf{y}} = {\left( {\frac{{\det \underline {\bf{\Sigma }} }}{{\det \overline {\bf{\Sigma }} }}} \right)^{\frac{1}{2}}}P_{\bm\mu ,\underline {\bf{\Sigma }} }^{\bf{b}}, \\
\overline{P_{{\bm\mu}, \mathbf{\Sigma}}^{{\bf b}}} =
\int\nolimits_{{\bf{y}} \in \left\{ {\left. {\bf{y}} \right|{\bf{b}} = {{\bf{b}}^{\beta \left( {\bf{y}} \right)}}} \right\}} \overline{p\left( {\bf{y}} \right)}\mathrm{d}{\bf{y}} = {\left( {\frac{{\det \overline {\bf{\Sigma }} }}{{\det \underline {\bf{\Sigma }} }}} \right)^{\frac{1}{2}}}P_{\bm\mu ,\overline {\bf{\Sigma }} }^{\bf{b}}.
\end{align}
Hence, \eqref{EqLemma41} is proven.

With \eqref{EqLemma41}, we derive the bounds of ${\frac{{P_{{\bm\mu}, \mathbf{\Sigma}}^{{\bf b}_t}}}{{P_{{\bm\mu}, \mathbf{\Sigma}}^{{\bf b}_e}}}}$ as
\begin{equation}\label{EqLemma4Proof3}
\frac{{P_{\bm\mu ,\underline {\bf{\Sigma }} }^{{{\bf{b}}_t}}\det \underline {\bf{\Sigma }} }}{{P_{\bm\mu ,\overline {\bf{\Sigma }} }^{{{\bf{b}}_e}}\det \overline {\bf{\Sigma }} }} = {\frac{\underline{P_{{\bm\mu}, \mathbf{\Sigma}}^{{\bf b}_t}}}{\overline{P_{{\bm\mu}, \mathbf{\Sigma}}^{{\bf b}_e}}}}  \le {\frac{{P_{{\bm\mu}, \mathbf{\Sigma}}^{{\bf b}_t}}}{{P_{{\bm\mu}, \mathbf{\Sigma}}^{{\bf b}_e}}}} \le {\frac{\overline{P_{{\bm\mu}, \mathbf{\Sigma}}^{{\bf b}_t}}}{\underline{P_{{\bm\mu}, \mathbf{\Sigma}}^{{\bf b}_e}}}} = \frac{{P_{\bm\mu ,\overline {\bf{\Sigma }} }^{{{\bf{b}}_t}}\det \overline {\bf{\Sigma }} }}{{P_{\bm\mu ,\underline {\bf{\Sigma }} }^{{{\bf{b}}_e}}\det \underline {\bf{\Sigma }} }}
\end{equation}
Then, \eqref{EqLemma42} is proven by introducing \eqref{EqLemma4Proof3} into \eqref{EqPEP}.
With \eqref{EqLemma41} and \eqref{EqLemma42}, \eqref{EqLemma43} and \eqref{EqLemma44} are derived directly.
\end{proof}

As shown in \cite{5398830}, we have ${\bf{P}}_\infty \le {\hat{\bf P}} \le {\bf M} = {\bf C}^{-1}{\bf V}{\bf C}^{-T}$ if ${\bf{P}}_\infty \le {\hat{\bf P}}\left[0\right]$ and $\gamma = 1$. With Lemma \ref{LemmaLoewner}, ${\bf\Sigma}_{S_i}$ of $S_i$ is bounded as $\underline{{\bf\Sigma}_{S_i}} \le {\bf\Sigma}_{S_i} \le \overline{{\bf\Sigma}_{S_i}}$ with
$\underline{{\bf\Sigma}_{S_i}} = g \circ h^{i+1} \left({{\mathbf{P}}_\infty }\right)$ and
$\overline{{\bf\Sigma}_{S_i}} = g \circ h^{i+1} \left({\bf M}\right)$.
With Lemma \ref{LemmaBoundsGivenSigmaBound}, $P_{{S_i}}^{{S_{i+1}}}$ is bounded in Lemma \ref{LemmaPSibounds}.

\begin{lemma}\label{LemmaPSibounds}
Given $\underline{{\bf\Sigma}_{S_i}} \le {\bf\Sigma}_{S_i} \le \overline{{\bf\Sigma}_{S_i}}$,
$P_{{S_i}}^{{S_{i+1}}}$ is bounded as
\begin{equation}
\underline{P_{{S_i}}^{{S_{i+1}}}}  \le P_{{S_i}}^{{S_{i+1}}} \le \overline{P_{{S_i}}^{{S_{i+1}}}},
\end{equation}
where
\begin{align}
\underline{P_{{S_i}}^{{S_{i+1}}}} &= \mathop {\min }\limits_{\bm\mu  \in \left[ { - Z,Z} \right)} \underline {P_{\bm\mu ,{{\bf\Sigma}_{S_i}}}^{{\rm{LB}}}}, \\
\overline{P_{{S_i}}^{{S_{i+1}}}} &= \mathop {\max }\limits_{\bm\mu  \in \left[ { - Z,Z} \right)} \overline {P_{\bm\mu ,{{\bf\Sigma}_{S_i}}}^{{\rm{UB}}}}.
\end{align}
\end{lemma}
\begin{proof}
By the definition of $P_{S_i}^{S_{i+1}}$, we have
\begin{equation}
\begin{aligned}
P_{S_i}^{S_{i+1}} &= \Pr\left(\left. {{{\bf{b}}_t} \ne \hat {\bf{b}}} \right|S_i\right) \\
&= \int\nolimits_{{\bm\mu}, {\bf\Sigma}_{S_i}} P_{{\bm\mu}, {\bf\Sigma}_{S_i}}^{{\rm{MAP}}}\cdot p\left(\left. {{\bm\mu}, {\bf\Sigma}_{S_i}}\right|S_i\right)\cdot \mathrm{d}{\bm\mu} \mathrm{d}{\bf\Sigma}_{S_i}  \\
\end{aligned}
\end{equation}
With Lemma \ref{LemmaBoundsGivenSigmaBound}, $P_{S_i}^{S_{i+1}}$ is bounded as
\begin{equation}
\begin{aligned}
P_{S_i}^{S_{i+1}} &\le \int\nolimits_{{\bm\mu}, {\bf\Sigma}_{S_i}} \overline {P_{\bm\mu ,{{\bf\Sigma}_{S_i}}}^{{\rm{UB}}}}\cdot p\left(\left. {{\bm\mu}, {\bf\Sigma}_{S_i}}\right|S_i\right)\cdot \mathrm{d}{\bm\mu} \mathrm{d}{\bf\Sigma}_{S_i} \\
&\le \mathop {\max }\limits_{\bm\mu  \in \left[ { - Z,Z} \right)} \overline {P_{\bm\mu ,{{\bf\Sigma}_{S_i}}}^{{\rm{UB}}}} = \overline{P_{{S_i}}^{{S_{i+1}}}},
\end{aligned}
\end{equation}
\begin{equation}
\begin{aligned}
P_{S_i}^{S_{i+1}} &\ge \int\nolimits_{{\bm\mu}, {\bf\Sigma}_{S_i}} \underline {P_{\bm\mu ,{{\bf\Sigma}_{S_i}}}^{{\rm{LB}}}}\cdot p\left(\left. {{\bm\mu}, {\bf\Sigma}_{S_i}}\right|S_i\right)\cdot \mathrm{d}{\bm\mu} \mathrm{d}{\bf\Sigma}_{S_i} \\
&\ge \mathop {\min }\limits_{\bm\mu  \in \left[ { - Z,Z} \right)} \underline {P_{\bm\mu ,{{\bf\Sigma}_{S_i}}}^{{\rm{LB}}}}= \underline{P_{{S_i}}^{{S_{i+1}}}}.
\end{aligned}
\end{equation}
Thus, Lemma \ref{LemmaPSibounds} is proven.
\end{proof}

By Lemma \ref{LemmaPSibounds}, \eqref{EqPMAPSi} is bounded in Theorem \ref{TheoremMAPMarkov}.
\begin{theorem}\label{TheoremMAPMarkov}
$P^{{\rm{MAP}}}$ is bounded as
\begin{equation}\label{EqTheorem2Bounds}
\frac{\underline{P_{{S_0}}^{{S_1}}}\underline\eta}{{1 + \overline{P_{{S_0}}^{{S_1}}}}\overline\eta } \le P^{{\rm{MAP}}}  \le \frac{\overline{P_{{S_0}}^{{S_1}}}\overline\eta}{{1 + \underline{P_{{S_0}}^{{S_1}}}}\underline\eta },
\end{equation}
where $\underline\eta = 1 + \sum\nolimits_{i = 1}^\infty  {\left( {\prod\nolimits_{j = 1}^i \underline{P_{{S_{j}}}^{{S_{j+1}}}} } \right)}$ and $\overline\eta = 1 + \sum\nolimits_{i = 1}^\infty  {\left( {\prod\nolimits_{j = 1}^i \overline{P_{{S_{j}}}^{{S_{j+1}}}} } \right)}$.
\end{theorem}

\section{Approximation of Bounds of MAP Decoding}

In this section, we first provide the upper bound $\overline P$ of $P_{{S_{i}}}^{{S_{i+1}}}$ and simplify Theorem \ref{TheoremMAPMarkov} as Theorem \ref{TheoremMAPApprox}.
Then, if $\overline P \ll 1$, the bounds in Theorem \ref{TheoremMAPApprox} are approximated in Corollary \ref{CorollaryMAPApproxPS0}.

\begin{theorem}\label{TheoremMAPApprox}
Given the upper bound $\overline P$ of $P_{{S_{i}}}^{{S_{i+1}}}$, i.e., $P_{{S_{i}}}^{{S_{i+1}}} \le \overline P$, $i=0,1,\cdots$, \eqref{EqTheorem2Bounds} is simplified to
\begin{equation}\label{EqTheorem2Bounds2}
\left(1 - \overline P\right)\underline{P_{{S_0}}^{{S_1}}} \le P^{{\rm{MAP}}} \le \left(1 - \overline P\right)^{-1}\overline{P_{{S_0}}^{{S_1}}}.
\end{equation}
\end{theorem}
\begin{proof}
Due to
\begin{align}
{1 + \overline {P_{{S_0}}^{{S_1}}}  \overline \eta } \le 1 + \sum\nolimits_{i = 1}^\infty  {{{\overline P}^i}}  = \left(1 - \overline P\right)^{-1},\\
\overline \eta \le 1 + \sum\nolimits_{i = 1}^\infty  {{{\overline P}^i}}  = \left(1 - \overline P\right)^{-1},
\end{align}
the lower bound of \eqref{EqTheorem2Bounds} is
\begin{equation}
{P^{{\rm{MAP}}}} \ge \frac{{\underline {P_{{S_0}}^{{S_1}}} \underline\eta }}{{1 + \overline {P_{{S_0}}^{{S_1}}}  \overline \eta }} \ge \frac{{\underline {P_{{S_0}}^{{S_1}}} }}{{1 + \overline {P_{{S_0}}^{{S_1}}} \overline \eta }} \ge \left( {1 - \overline P} \right)\underline {P_{{S_0}}^{{S_1}}} ,
\end{equation}
and the upper bound is
\begin{equation}
{P^{{\rm{MAP}}}} \le \frac{{\overline {P_{{S_0}}^{{S_1}}} \overline\eta }}{{1 + \underline {P_{{S_0}}^{{S_1}}}  \underline \eta }} \le {\overline {P_{{S_0}}^{{S_1}}} \overline\eta } \le \left(1 - \overline P\right)^{-1}\overline{P_{{S_0}}^{{S_1}}}.
\end{equation}
Thus, Theorem \ref{TheoremMAPApprox} is proven.
\end{proof}

\begin{remark}
As $i$ goes to infinity, the trace of ${{\bf\Sigma}_{S_i}}$ goes to infinity, which makes $P_{{S_i}}^{{S_{i + 1}}}$ approaches to the ML performance $P^{\rm ML}$, i.e., $\mathop {\lim }\limits_{i \to \infty } P_{{S_i}}^{{S_{i + 1}}} = {P^{{\rm{ML}}}}$. With this, we set $\overline P$ as the union bound of $P^{\rm ML}$, i.e.,
$
\overline P = \sum\limits_{{{\bf{b}}_e} \in {{\left\{ {0,1} \right\}}^{{N_y}n}}\backslash {{\bf{b}}_t}} Q\left( {\sqrt {\frac{{2{d_{{{\bf{c}}_t},{{\bf{c}}_e}}}}}{{{N_0}}}} } \right)
$
\end{remark}

\begin{corollary}\label{CorollaryMAPApproxPS0}
If $\overline P \ll 1$, we have
$\underline{P_{{S_0}}^{{S_1}}}\lesssim P^{{\rm{MAP}}} \lesssim \overline{P_{{S_0}}^{{S_1}}}.
$
\end{corollary}

\section{Simulation Results}

\begin{figure}[t]
\setlength{\abovecaptionskip}{0.cm}
\setlength{\belowcaptionskip}{-0.cm}
  \centering{\includegraphics[scale=0.59]{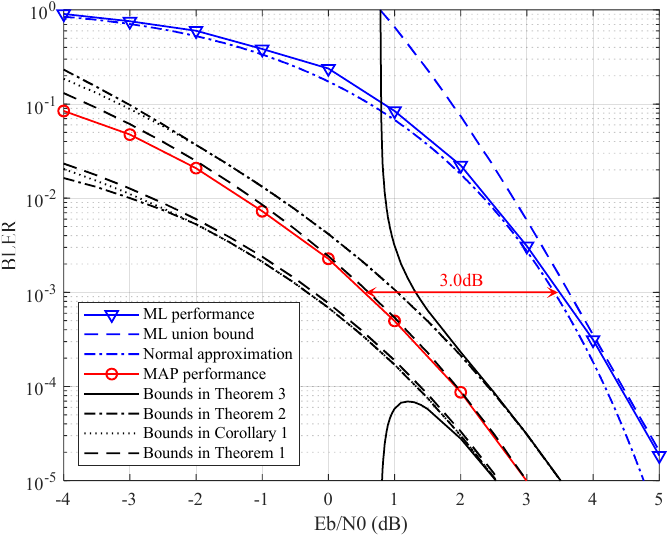}}
  \caption{BLER performance and bounds of scalar system.}\label{FigBLERNy1}
  \vspace{-1em}
\end{figure}

\begin{figure}[t]
\setlength{\abovecaptionskip}{0.cm}
\setlength{\belowcaptionskip}{-0.cm}
  \centering{\includegraphics[scale=0.59]{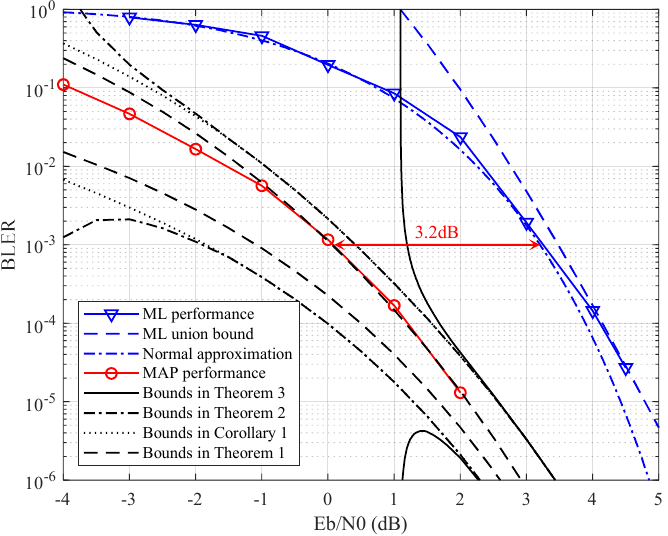}}
  \caption{BLER performance and bounds of vector system.}\label{FigBLERNy2}
  \vspace{-1em}
\end{figure}

\begin{figure}[t]
\setlength{\abovecaptionskip}{0.cm}
\setlength{\belowcaptionskip}{-0.cm}
  \centering{\includegraphics[scale=0.59]{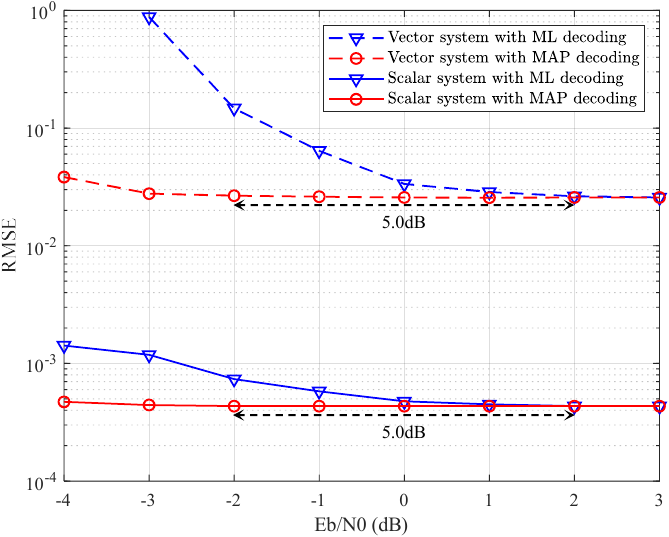}}
  \caption{RMSE performance of scalar and vector systems.}\label{FigRMSE}
  \vspace{-1em}
\end{figure}

We first provide the BLER performance bounds of scalar system and vector system with MAP decoding in Fig. \ref{FigBLERNy1} and Fig. \ref{FigBLERNy2}, respectively. Then the RMSE performance is shown in Fig. \ref{FigRMSE}. Polar codes constructed by the polar sequence with 16-bit CRC in 5G are used as the coding schemes \cite{3GPP_5G_polar} and the normal approximation of the finite block rate is provided \cite{5452208}. We calculate the first 5 terms of $\underline\eta$ and $\overline\eta$ to calculate the bounds in Theorem \ref{TheoremMAPMarkov}.

Fig. \ref{FigBLERNy1} illustrates the BLER performance and bounds of scalar system with MAP decoding, where the system model parameters are $A=-1.13$, $B=1$, $C=1$, $W = V = 10^{-7}$ and $x_{\rm{ref}}\left[ k \right] = 0, k=1,2,\cdots$, the quantization parameters are $Z = 5$ and $n = 16$, and the polar code parameters are $N = 64$ and $K = N_yn = 16$ with 16-bit CRC.
In Fig. \ref{FigBLERNy1}, we can first observe that the MAP performance outperforms the ML performance and the normal approximation and shows about $3.0$dB performance gain at BLER $10^{-3}$. Thus, the joint detection with KF and channel decoding can break through the limit of finite block rate and exhibits potential in communications with high real-time and reliability by short code length. Then, the bounds in Theorem \ref{TheoremMAPlimit} are closer to the MAP performance, compared with the bounds in Theorem \ref{TheoremMAPMarkov}, Theorem \ref{TheoremMAPApprox} and Corollary \ref{CorollaryMAPApproxPS0}.
The MAP performance coincides with the upper bound in Theorem \ref{TheoremMAPlimit} as the SNR increases and is larger than the lower bound in Theorem \ref{TheoremMAPlimit}.
Thus, the upper bound in Theorem \ref{TheoremMAPlimit} can be used to evaluate the MAP performance in the high SNR region.
The bounds in Theorem \ref{TheoremMAPMarkov} and Theorem \ref{TheoremMAPApprox} are close to these in Corollary \ref{CorollaryMAPApproxPS0} as the SNR increases.
For the bounds in Theorem \ref{TheoremMAPApprox}, since we set $\overline P$ as ML union bound, the SNR limit of bounds is identical to the SNR of ML union bound at BLER $10^0$.

Fig. \ref{FigBLERNy2} illustrates the BLER performance and bounds of vector system with MAP decoding, where the system model parameters are ${\bf{A}} = \left[\begin{smallmatrix} {1.25}&0\\ 1&{1.1} \end{smallmatrix}\right]$,
${\bf{B}} = \left[ {\begin{smallmatrix} 1\\ 1 \end{smallmatrix}} \right]$,
${\bf{C}} = \left[ {\begin{smallmatrix}1&0\\0&1\end{smallmatrix}} \right]$,
${\bf{W}} = {\bf{V}} = \left[ {\begin{smallmatrix}{{{10}^{ - 4}}}&0\\0&{{{10}^{ - 4}}}\end{smallmatrix}} \right]$, and
${\bf x}_{\rm{ref}}\left[ k \right] = \left[0, 0 \right]^T, k=1,2,\cdots,$
the quantization parameters are $Z = 5$ and $n = 10$, and the polar code parameters are $N = 64$ and $K = N_yn = 20$ with 16-bit CRC. In Fig. \ref{FigBLERNy2}, we can obtain the same conclusion as Fig. \ref{FigBLERNy1} and the MAP performance shows about $3.2$dB performance gain at BLER $10^{-3}$ compared with the ML performance and the normal approximation.

Fig. \ref{FigRMSE} shows the RMSE performance of scalar and vector systems with MAP decoding and the parameters are identical to Fig. \ref{FigBLERNy1} and Fig. \ref{FigBLERNy2}.
We can observe that both the systems with MAP decoding have about $5.0$dB performance gain compared with those with ML decoding when the RMSE converges, which shows the advantage of the joint design of controls and communications with prior information.

\section{Conclusion}

In this paper, we  regard the joint detection with KF and channel decoding as MAP decoding and derive the corresponding bounds. We first derive the limiting bounds as the SNR goes to infinity and the system interference goes to zero. Then, we construct an infinite-state Markov chain and derive the MAP bounds. Finally, we approximate the MAP bounds as the bounds of the transition probability from state $S_0$ to state $S_1$.
The simulation results show that the MAP performance coincides with the limiting upper bound as the SNR increases and outperforms the normal approximation.

\bibliographystyle{IEEEtran}
\bibliography{IEEEabrv,myrefs}

\end{document}